\documentclass{article}

\usepackage{hyperref}
\usepackage{amsmath}
\usepackage{amsthm}
\usepackage{amssymb}
\usepackage{mathtools}
\usepackage{todonotes}
\usepackage{setspace}
\usepackage{breqn}
\usepackage{fullpage}
\usepackage{authblk}

\newtheorem{theorem}{Theorem}
\newtheorem{lemma}{Lemma}
\newtheorem{definition}{Definition}
\newtheorem{remark}{Remark}

\usepackage[ruled]{algorithm2e}

\let\oldnl\nl
\newcommand{\nonl}{\renewcommand{\nl}{\let\nl\oldnl}}

\def \ie {\textit{i.e.}}
\def \eg {\textit{e.g.}}
\def \etal {\textit{et~al.}}

\def \cS {\mathcal S}

\DeclareMathOperator* {\argmin} {arg\,min}

\newcommand{\N}{\mathbb{N}}
\newcommand{\R}{\mathbb{R}}
\newcommand{\et}{\tau}
\newcommand{\goal}{z}
\newcommand{\opt}{\et^*}
\newcommand{\scv}{\kappa}

\newcommand{\vs}{i_S}
\newcommand{\seq}{\gamma}
\newcommand{\Seq}{\Gamma}

\newcommand{\tw}{\theta}
\newcommand{\dm}{\delta}
\newcommand{\bigo}{\mathcal{O}}

\begin{document}

\title{Computational Aspects of Optimal Strategic Network Diffusion}

\author[a,b,*]{Marcin Waniek}
\author[b]{Khaled Elbassioni}
\author[c,d]{Fl\'avio L. Pinheiro}
\author[d]{\\C\'esar A. Hidalgo}
\author[b]{Aamena Alshamsi}

\renewcommand*{\Affilfont}{\normalsize}

\affil[a]{Compute Science, New York University Abu Dhabi, Abu Dhabi, UAE}
\affil[b]{Masdar Institute, Khalifa University of Science and Technology, Abu Dhabi, UAE}
\affil[c]{Nova Information Management School (NOVA IMS), Universidade Nova de Lisboa, Lisboa, Portugal}
\affil[d]{MIT Media Lab, Massachusetts Institute of Technology, Cambridge, USA}
\affil[*]{To whom correspondence should be addressed:  mjwaniek@gmail.com}

\date{}

\maketitle

\begin{abstract}
Diffusion on complex networks is often modeled as a stochastic process.
Yet, recent work on strategic diffusion emphasizes the decision power of agents \cite{alshamsi2018optimal} and treats diffusion as a strategic problem.  
Here we study the computational aspects of strategic diffusion, \ie, finding the optimal sequence of nodes to activate a network in the minimum time.
We prove that finding an optimal solution to this problem is NP-complete in a general case.
To overcome this computational difficulty, we present an algorithm to compute an optimal solution based on a dynamic programming technique.
We also show that the problem is fixed parameter-tractable when parametrized by the product of the treewidth and maximum degree.
We analyze the possibility of developing an efficient approximation algorithm and show that two heuristic algorithms proposed so far cannot have better than a logarithmic approximation guarantee.
Finally, we prove that the problem does not admit better than a logarithmic approximation, unless P=NP.
\end{abstract}

\section{Introduction}
\label{sec:introduction}

% 1. STOCHASTIC DIFFUSION MODELS
Diffusion in social networks has been a widely studied application in problems such as epidemic outbreaks and  the adoption of behaviors and innovations \cite{valente1995network, rogers2010diffusion, bailey1975mathematical, pastor2001epidemic, aral2017exercise, vasconcelos2019consensus}. The approach to model diffusion can be roughly divided between simple and complex contagion processes. Under simple contagion transmission requires contact with one individual at a constant rate such as transmission of infectious diseases \cite{kermack1927contribution}. An example of simple contagion process is the independent cascade model~\cite{kempe2003maximizing}. In contrast, complex contagion requires reinforcement from multiple sources. An example of complex contagion process is the linear threshold model~\cite{kempe2003maximizing}. Complex contagion has been widely studied in the context of cascade phenomena, with particular applications to viral marketing campaigns \cite{goldenberg2001using,leskovec2007dynamics,domingos2001mining}. In that context, the problem under analysis is that of finding the initial set of seeds that would maximize diffusion.
%Despite the differences, both of these approaches study a purely stochastic dynamics of diffusion.
Different approaches to solving the seed selection problem include using centrality measures~\cite{kiss2008identification}, network percolation~\cite{morone2015influence}, and propagation traces~\cite{goyal2011data}, among many others~\cite{chen2009efficient,hinz2011seeding,libai2013decomposing}.
Some works in the literature take an adaptive approach to the topic, basing the selection of seeds on additional knowledge gathered during the process, be it either the current state of a dynamic network topology~\cite{tong2017adaptive}, or expanded pool of potential seeds~\cite{seeman2013adaptive,horel2015scalable}.
An alternative way of increasing the diffusion coverage is spreading available seeds over time~\cite{jankowski2017balancing,jankowski2018strategic}.
Despite the strategic aspects of different methods of seed selection, the diffusion process itself is purely stochastic.

% 2. MOVE TO STRATEGIC DIFFUSION, MOTIVATE WHY IT IS IMPORTANT

However, there are cases in which diffusion takes place according to the same rules of complex contagion, but that are not purely stochastic (\ie, they have a strategic component) and in which the network itself is not, necessarily, a social network. Take for instance how regions develop new economic activities \cite{hidalgo2007product} or start production in new research fields \cite{guevara2016research}. In these cases agents are not embedded in the network, instead they are taking actions over a networked system, which captures the relatedness between economic or academic activities. More importantly, not only the choice of the initial state is strategic, but also the whole process of diffusion is strategic. %Indeed, economic agents and policy makers often design development strategies that target specific areas of socio-economic and innovation systems.
Stochasticity in this context is not at the pairwise interactions between agents, but in the success rate of agents in entering new activities. In case of the economic development, the probability of success while entering a new activity (represented by a node in the network of related economic activities) increases with the number of other already developed activities that are connected to it. Thus, agents can fail or succeed depending on the particular context of each strategic action.

Alshamsi~\etal~\cite{alshamsi2018optimal} proposed to study the diversification of economic activities through the lenses of a strategic diffusion model. In this model the entire process of diffusion is guided by a strategic agent. The agent's actions concern the selection of a node in a network to be targeted at each time step. Hence, the node can become activated in the next step of the action sequence of the agent, or not. Following existing empirical findings \cite{hidalgo2007product,pinheiro2018shooting}, the model assumes that the activation probability (\ie, the probability of success of the agent's action) is captured by a complex contagion process. The goal is to activate the entire network, starting from a single active node, in the minimum time possible. While Alshamsi~\etal~\cite{alshamsi2018optimal} showed for a wide range of network topologies that activating a network in an optimal manner requires a balance between exploitation and exploration strategies, many questions regarding the computational feasibility of finding optimal strategies remained open. Here, we explore several theoretical computational considerations of strategic diffusion processes.

\subsection{Results and Methods}

We start by exploring the computational feasibility of finding an optimal strategy that minimizes the total diffusion time, \ie, answering the question of what is the best sequence of activating a given number of nodes in the network (starting from a seed node) in the shortest expected time, while taking into consideration that the time necessary to activate each node is proportional to the number of active neighbors.

We show that the decision version of the problem is NP-complete (Theorem~\ref{thrm:npc}). We prove it using a reduction from the Set Cover problem. This observation indicates that developing a polynomial algorithm for finding an optimal sequence of strategic diffusion is not realistic, unless P=NP.

Given this difficulty, we propose an algorithm for computing an optimal way of activating a given number of nodes in the shortest time possible (Algorithm~\ref{alg:dynamic-programming}). While the algorithm takes exponential time to find the solution, it still outperforms simple exhaustive search and can be used to compute an optimal solution for networks of moderate size. The algorithm utilizes the dynamic programming technique to find the fastest way to reach every possible state of network activation.

We also present two algorithms finding an optimal solution for a more restrictive class of networks, \ie, networks with bounded treewidth and bounded maximal degree. They traverse a tree decomposition of a network and find an optimal way of activating either an entire network or a given number of nodes, respectively. As a consequence, we show that the problem is fixed parameter-tractable when parametrized by the sum of the treewidth and maximum degree.

As finding an optimal solution proves to be computationally demanding, we turn our attention to assessing the possibility of obtaining an efficient approximation algorithm. We first investigate whether two effective heuristic algorithms proposed by Alshamsi~\etal~\cite{alshamsi2018optimal} have constant approximation ratios. One of them is the greedy algorithm, \ie, always targeting the node with the highest probability of activation. The other is the majority algorithm, \ie, targeting the node with the highest number of active neighbors. We show that the approximation ratio for both of these algorithms is $\Omega(\log n)$ (Theorems~\ref{thrm:greedy} and~\ref{thrm:majority}).
The proofs are based on constructing a sequence of networks where the ratio between total expected time of activation of the solution obtained using the heuristic algorithm and the optimal expected time of activation goes to infinity.

Finally, we show that, unless P=NP, there is no way to approximate the problem within a ratio better than $\ln n$, in particular it is impossible to construct an $r$-approximation algorithm for a constant $r$ (Theorem~\ref{thrm:logn}). We prove this claim by showing a reduction from the Minimum Set Cover problem and using the fact that Minimum Set Cover cannot be approximated within a ratio of $(1-\epsilon) \ln n$ for any $\epsilon > 0$, unless P=NP~\cite{dinur2014analytical}.

\subparagraph*{Organization of the manuscript}

The remainder of the article is organized as follows. Section~\ref{sec:preliminaries} describes the notations and computational problems used in our reductions. Section~\ref{sec:problem-definition} presents a formal definition of the strategic diffusion process and the main problem considered in our study. In Section~\ref{sec:optimal-solution} we present our hardness result for the decision version of the problem and we describe a dynamic programming algorithm to find an optimal way of strategic diffusion. Section~\ref{sec:treewidth} introduces algorithms computing optimal solution for networks with bounded treewidth and maximal degree. Section~\ref{sec:approximation} describes our results concerning approximation of the optimal solution. Section~\ref{sec:conclusions} presents conclusions and potential ideas for future work.

\section{Preliminaries \& Notation}
\label{sec:preliminaries}

In this section, we present notations and concepts that will be used throughout the paper.

\subsection{Basic Network Notation}

Let $G = (V,E,W)$ denote a network with weighted edges, where $V=\{1,\ldots,n\}$ denotes the set of $n$ nodes, $E \subseteq V \times V$ denotes the set of edges and $W \in \R^{n \times n}$ denotes the matrix with weights of edges.
We denote an edge between nodes $i$ and $j$ by $ij$.
In this work we consider networks that are \textit{undirected}, \ie, we do not discern between edges $ij$ and $ji$.
We also assume that networks do not contain self-loops, \ie, $\forall_{i \in V}ii \notin E$.
We denote by $N_G(i)$ the set of \textit{neighbors} of $i$ in $G$, \ie, $N_G(i) = \{j \in V : ij \in E\}$.
We denote by $d_G(i)$ the \textit{degree} of $i$ in $G$, \ie, $d_G(i) = |N_G(i)|$.

We consider networks with weighted edges.
We denote by $w_{ij} \in \R$ the weight of the connection from $i$ to $j$, we will call this value \textit{influence} that $i$ has on $j$.
We do not assume that the relation of influence is symmetric, \ie, it is possible that for $ij \in E$ we have $w_{ij} \neq w_{ji}$.
Unless stated otherwise, we will assume that $\forall_{i,j \in V} w_{ij} \geq 0$.
We also assume that if $ij \notin E$ then $w_{ij}=0$.
We denote by $w_i$ the sum of influence on node $i$, \ie, $w_i=\sum_{j \in N(i)}w_{ji}$.
We will typically assume that $\forall_{i \in V} w_i > 0$.

Let $\Seq(V)$ denote the set of all \textit{ordered sequences} of elements from $V$ without repetitions. 
Let $\seq_i$ denote the $i$-th element (node) of sequence $\seq \in \Seq(V)$, and $|\seq|$ denote the number of elements in $\seq \in \Seq(V)$.
Finally, we call  $\seq \in \Seq(V)$ a {\it full} sequence if  $|\seq|=|V|$ and we denote the set of full sequences by $\Seq^*(V)$.

To make the notation more readable, we will often omit the network itself from the notation when it is clear from the context, \eg, by writing $N(i)$ instead of $N_G(i)$. We sometimes treat sequences as sets, when the order is not important. We use $\oplus$ to denote the concatenation operation over sequences.

\subsection{Computational Problems}

In our reductions we will use two standard versions of the Set Cover problem, decision and combinatorial optimization.

\begin{definition}[Set Cover~\cite{karp1972reducibility}]
An instance of the Set Cover problem is defined by a universe $U=\{u_1, \ldots, u_{|U|}\}$, a collection of sets $\cS = \{S_1, \ldots, S_{|\cS|}\}$ such that $\forall_j{S_j \subset U}$, and an integer $k \leq |\cS|$.
The goal is to determine whether there exist $k$ elements of $\cS$ the union of which equals $U$.
\end{definition}

Set Cover is one of the classic $21$ Karp's NP-complete problems.

\begin{theorem}[\cite{karp1972reducibility}]
\label{thrm:set-cover-npc}
Set Cover problem is NP-complete.
\end{theorem}

For the proof of the approximation hardness we will use the minimization version of the problem.

\begin{definition}[Minimum Set Cover]
An instance of the Minimum Set Cover problem is defined by a universe $U=\{u_1, \ldots, u_{|U|}\}$ and a collection of sets $S = \{S_1, \ldots, S_{|S|}\}$ such that $\forall_j{S_j \subset U}$.
The goal is to find subset $S^* \subseteq S$ such that the union of $S^*$ equals $U$ and the size of $S^*$ is minimal.
\end{definition}

For $\alpha\ge1$, an $\alpha$-approximation for a given instance of a minimization problem is a feasible solution whose objective is within a factor of $\alpha$ of any optimal solution. We will use the fact that Minimum Set Cover problem is hard to approximate.

\begin{theorem}[\cite{dinur2014analytical}]
\label{thrm:set-cover-ptas}
For any fixed $\epsilon > 0$ Minimum Set Cover cannot be approximated to within $(1-\epsilon) \ln n$, unless P=NP.
\end{theorem}

We now move to defining the model of strategic diffusion and the main optimization problem of our study.

\section{Problem Definition}
\label{sec:problem-definition}

In this section we describe the process of strategic network diffusion, which is the main focus of our study, as well as the computational problem concerning it.

Most diffusion models are purely stochastic~\cite{goldenberg2001using,kempe2003maximizing}.
In this work however, we focus on the strategic model of diffusion to account for cases in which the interconnections between targets affect their activation time and therefore choosing the order in which nodes will be targeted for activation is strategically planned. 

In this model, the process of diffusion is driven by a strategic agent.
At the beginning of the process only one chosen node of the network, the \textit{seed node} $\vs$, is active.
Then, the agent chooses a sequence $\seq \in \Seq(V)$ that provides the order in which the nodes will be activated.
The probability of successful activation of a node $i$ in one attempt is given by:
$$
p(i) = \beta \left( \frac{\sum_{j \in N(i) \cap A} w_{ji}}{w_i} \right)^\alpha
$$
where $A$ is the set of currently active nodes (at the beginning of the process it consists only of the seed), and $\alpha,\beta \in [0,1]$ are constants.
Unless stated otherwise, we will assume that $\alpha = \beta = 1$.
Notice that the expected time of activation of node $i$ with non-zero activation probability is $\et(i)=\frac{1}{p(i)}$.
By $\et(\seq)$ we will denote the expected time of activation of all nodes in the sequence $\seq$.

This model was proposed by Alshamsi~\etal~\cite{alshamsi2018optimal} for undirected networks with unweighted edges.
If we assume that $w_{ij}=1 \iff ij \in E$ then our model is exactly equivalent to the model proposed by Alshamsi~\etal~\cite{alshamsi2018optimal}.

We now define the main computational problem of our study.

\begin{definition}[Optimal Partial Diffusion Sequence]
This problem is defined by a tuple $(G,\vs,\goal)$, where $G=(V,E,W)$ is a given network with weighted edges, $\vs \in V$ is the seed node and $\goal \leq n$ in the number of nodes to activate.
The goal is to identify $\seq^* \in \Seq(V)$ such that $\seq^*_1 = \vs$, $|\seq^*|=\goal$ and $\et(\seq^*)$ is minimal.
\end{definition}
In other words, we intend to find the fastest way to activate $\goal$ nodes in the network. When $\goal=|V|$ in the above definition, the problem is simply called Optimal (Full) Diffusion Sequence.
% * <khaled.elbassioni@ku.ac.ae> 2018-06-03T09:07:27.123Z:
% 
% We need to prove that the simple version is NP-complete.
% 
% ^ <khaled.elbassioni@ku.ac.ae> 2018-06-03T09:07:53.113Z.

\begin{remark}\label{r1}
In the case of integer weights of {\it polynomial length}, for any instance of the Optimal Diffusion Sequence problem there exists a polynomially equivalent instance with binary weights.
\end{remark}

\begin{proof}
Let $(G,\vs,|V|)$ be a given instance of the Optimal Diffusion Sequence.
In order to construct an equivalent instance with $0/1$-weights we replace every edge $ij$ with a set of $w_{ij}$ parallel $2$-paths $\{\langle i,k_{i,j,1},j\rangle, \ldots,$ $\langle i,k_{i,j,w_{ij}},j\rangle\}$, and we set weights of the new edges to $w_{i k_{i,j,r}}=w_{k_{i,j,r} j}=1$,  $w_{k_{i,j,r} i}=w_{j k_{i,j,r}}=0$, for $r=1,\ldots,w_{ij}$.

Given a feasible solution $\seq$ to the original instance, we obtain a feasible solution $\seq'$ to the constructed instance by activating all nodes $k_{i,j,r}$ immediately after activating node $i$ in the original sequence (notice that the expected time of activation of every such node $k_{i,j,r}$ is $1$).
Thus, if the value of the solution $\seq$ is $t$, then the value of the $\seq'$ is $t+\sum_{ij\in E}\big(w_{ij}+w_{ji}\big)$.

Given a feasible solution $\seq'$ to the constructed instance of the binary version of the problem, we can obtain a solution with a lesser or equal value by activating all nodes $k_{i,j,r}$ in one block immediately after activating node $i$ (again, the expected time of activation of node $k_{i,j,r}$ is always $1$, while it contributes to the time of activation of node $j$).
For such sequence, we can obtain corresponding solution $\seq$ to the original instance of the problem by removing from it all nodes $k_{i,j,r}$.
If the value of the solution $\seq'$ is $t$, then the value of the $\seq$ is $t-\sum_{ij \in E}\big(w_{ij}+w_{ji}\big)$.

Note that this reduction does not work in general for the Optimal Partial Diffusion Sequence if $z<|V|$.
\end{proof}

\section{Computing an Optimal Solution}
\label{sec:optimal-solution}

We now describe our results on finding an efficient way to compute an exact solution to the Optimal Diffusion Sequence problem.

\subsection{Hardness of Finding an Optimal Solution}

First, we show that the decision version of Optimal Partial Diffusion Sequence problem is NP-complete in a general case.

\begin{theorem}
\label{thrm:npc}
Optimal Partial Diffusion Sequence problem is NP-complete, even if all weights are in $\{0,1\}$.
\end{theorem}

\begin{proof}
The decision version of the Optimal Partial Diffusion Sequence problem is the following: given a network $G=(V,E,W)$, the seed node $\vs$, the number of nodes to activate $\goal$, and a value $t^* \in \R^+$, does there exist a sequence of nodes $\seq^* \in \Seq(V)$ starting with $\vs$ such that $|\seq^*| = \goal$ and $\et(\seq^*) \leq t^*$?

This problem clearly is in NP, as given a solution, \ie, a sequence of nodes $\seq^* \in \Seq(V)$, we can compute its expected time of activation in polynomial time.

To prove the NP-hardness of the problem we will show a reduction from the NP-complete Set Cover problem.

\begin{figure}[tbh]
\centering
\includegraphics[width=.7\linewidth]{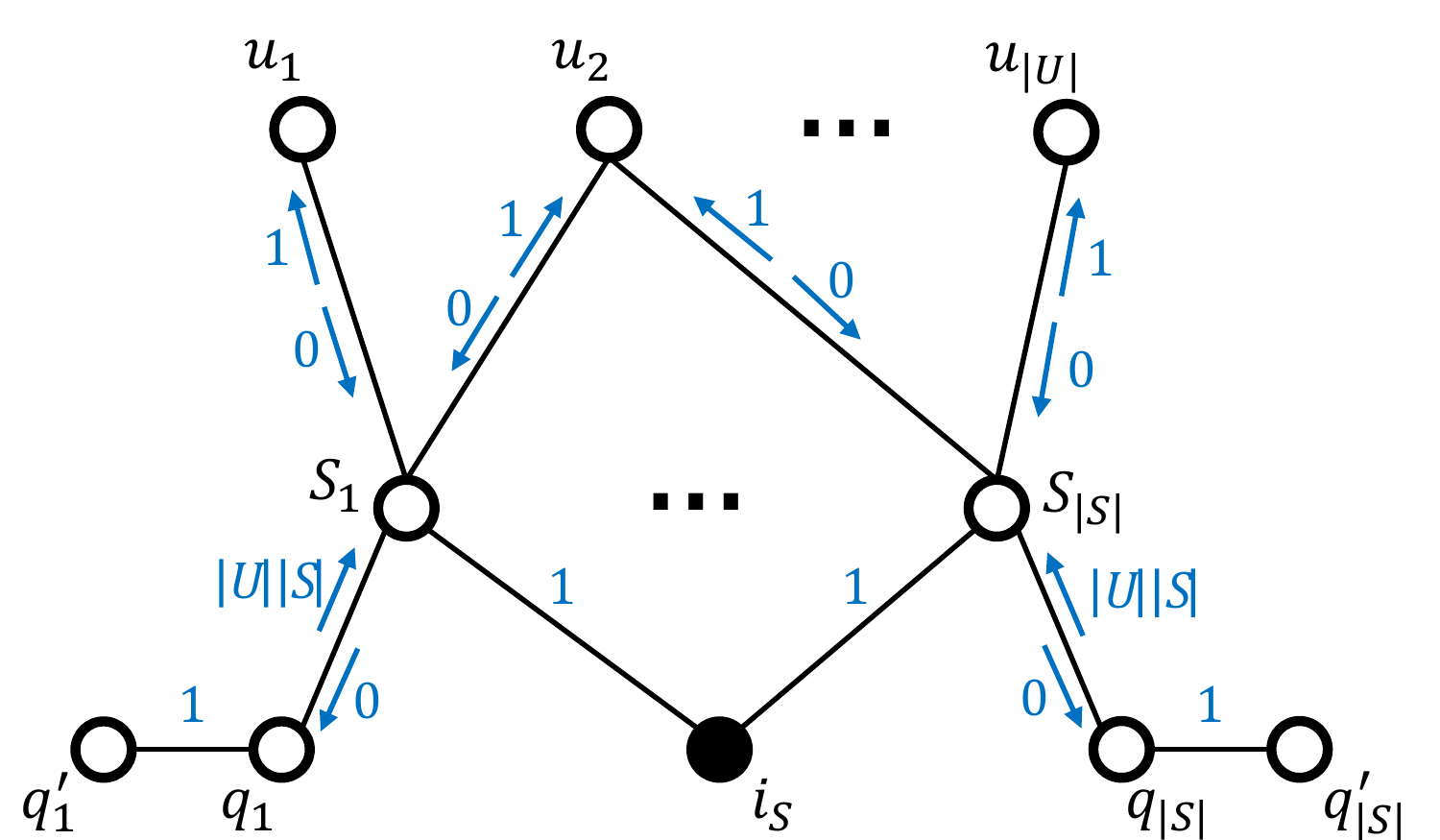}
\caption{Network $G$ constructed for the proof of Theorem~\ref{thrm:npc}.
Blue numbers next to edges express their weights.
If there are no arrows next to edge $ij$ then $w_{ij}=w_{ji}$.
Otherwise weight $w_{ij}$ is denoted next to the arrow pointing towards node $j$, and weight $w_{ji}$ is denoted next to the arrow pointing towards node $i$.}
\label{fig:nph}
\end{figure}

Let $(U,\cS,k)$ be an instance of the Set Cover problem.
We define network $G$ as follows (an example of such network is presented in Figure~\ref{fig:nph}):
\begin{itemize}
\item \textbf{The set of nodes:}
For every $S_i \in\cS$, we create three nodes, denoted by $S_i$, $q_i$ and $q'_i$.
For every $u_i \in U$, we create a single node $u_i$.
Additionally, we create a single node $\vs$.
\item \textbf{The set of edges:}
For every node $S_i$, we create an edge $S_i \vs$.
For every node $q_i$, we create edges $q_i S_i$ and $q_i q'_i$.
Finally, for every node $u_i$ and every $S_j$ such that $u_i \in S_j$ we create an edge $u_i S_j$.
\item \textbf{The weight matrix:}
For every edge $u_i S_j$ we set weights to $w_{u_i S_j}=0$ and $w_{S_j u_i}=1$.
For every edge $S_i q_i$ we set weights to $w_{S_i q_i}=0$ and $w_{q_i S_i}=|U||\cS|$.
For all other pairs of nodes connected with an edge we set the weights to $1$.
\end{itemize}

Let $\goal=k+|U|+1$ (we intend that the solution sequence will activate node $\vs$, $k$ nodes in $\cS$ and all nodes in $U$), and a value $t^*=k(|U||\cS|+1)+|U||\cS|$.
Now, consider an instance of the Optimal Partial Diffusion Sequence problem $(G,\vs,\goal)$.
We will now show that a solution to this instance  of value at most $t^*$ corresponds to a solution to the given instance of the Set Cover problem.

Notice that the expected time of activation of every node $S_i$ is always $|U||\cS|+1$, while for the expected time of activation of a node $u_i$ we have $\et(u_i) \leq |\{S_j : u_i \in S_j\}| \leq |\cS|$.
Notice also that neither any of the nodes $q_i$ nor any of the nodes $q'_i$ can be activated when $\vs$ is the seed node, as the influence of $S_i$ on $q_i$ is zero.

First, we will show that if there exists a solution $\cS^*$ to the given instance of the Set Cover problem, then there also exists a solution to the constructed instance of the Optimal Partial Diffusion Sequence problem of value at most $t^*$.
We can construct such a solution by first activating every node $S_i \in \cS^*$ ($k$ nodes activated in expected time $|U||\cS|+1$ each) and then activating all nodes $u_i$ ($|U|$ nodes activated in expected time not exceeding $|\cS|$ each).
Such $\seq^*$ activates $k+|U|$ nodes in expected time not exceeding $k(|U||\cS|+1)+|U||\cS|$, therefore it is a solution to the constructed instance of the Optimal Partial Diffusion Sequence problem of value at most $t^*$.

Second, to complete the proof of the NP-hardness, we have to show that if there exists a solution $\seq^*$ to the constructed instance of the Optimal Partial Diffusion Sequence problem of value at most $t^*$, then there also exists a solution to the given instance of the Set Cover problem.
Such a solution is $\cS^* = \seq^* \cap\cS$, \ie, choosing sets $S_i$ corresponding to nodes $S_i$ occurring in sequence $\seq^*$.
Notice that there cannot be more than $k$ such nodes, as activating $k+1$ nodes $S_i$ has expected time $(k+1)(|U||\cS|+1) > t^*$.
Since this is the case, in order to activate $k+|U|$ nodes other than $\vs$, sequence $\seq^*$ has to activate all nodes in $U$.
However, to activate a node $u_i$, we first have to activate at least one of its neighbors, \ie, node $S_j$ such that $u_i \in S_j$.
Therefore, for every node $u_i$ there must exist at least one node $S_j \in \seq^* \cap S$ such that $u_i \in S_j$.
Hence, $S^* = \seq^* \cap \cS$ is a valid solution to the given instance of the Set Cover problem.

Finally, we can use the construction in Remark~\ref{r1} to replace every edge $q_i S_i$  by a set of parallel paths with edge weights in $\{0,1\}$. Note that such a replacement does not change the value of the objective as the nodes $q_i$, and hence the intermediate nodes added   on the parallel paths,  are never activated.
This concludes the proof.
\end{proof}

Therefore, there exists no polynomial algorithm finding the optimal way to activate a given number of nodes in the process of strategic diffusion, unless P=NP.
However, we now propose an exponential algorithm based on dynamic programming technique.
%However, we now propose an exponential algorithm that is more efficient than simple exhaustive search.

\subsection{Dynamic Programming Algorithm}

We present an algorithm for computing an optimal solution to the problem using the dynamic programming technique~\cite{bellman1954theory}.
The main idea behind dynamic programming is breaking the main problem into a number of smaller, easier to solve sub-problems in a recursive manner.
However, unlike in standard recursion where often the same computation is repeated multiple times, in dynamic programming the solution to each sub-problem is computed only once and stored in memory for future use.
Over the years the dynamic programming techniques were developed in various directions~\cite{bertsekas1995dynamic,bertsekas1996neuro,rust2016dynamic,bellman2015applied}, however our algorithm is based on the original version of the technique. 

Notice that in our setting the activation probability of any given node depends on the set of active nodes in the network, however it does not depend on the order in which these nodes were activated.
Hence, the solution for each set of active nodes have to be computed only once.
Moreover, solving the sub-problem of finding an optimal way of activating $k$ nodes in the network allows to efficiently find the solution for $k+1$ nodes, as the total expected time of activation can be expressed in a recursive manner as a sum of time of activation of the initial $k$ nodes and the time of activation of the last node.
The idea of using solutions for smaller sub-problems to solve a larger problem is the core concept behind the dynamic programming, hence we find it a suitable optimization method for solving the Optimal Partial Diffusion Sequence problem.

The algorithm is based on the following recurrence relation, where $\et^*(C)$ is the minimal expected time of activation of a set of nodes $C$:

$$
\et^*(C) =
	\begin{cases}
		0 & \text{if}\ C=\{\vs\} \\
		\min\limits_{i \in C}\et^*(C \setminus \{i\}) + \frac{w_i}{\sum_{j \in N(i) \cap C}w_{ji}} & \text{if}\ |C| > 1 \\
		\infty & \text{otherwise}
	\end{cases}
$$
where we assume that summation over empty set results in zero and division by zero results in $\infty$.

Out of sets of size one, only the set consisting of the seed node has finite time of activation.
For sets that are larger than one we divide the problem into computing the optimal time of activation of all nodes but the last one, and then computing the time of activation of the last node.
Since in the strategic diffusion process the nodes are activated sequentially, one of the nodes $i$ from the set $C$ has to be activated last, and we are able to compute its time of activation based on the information of which other nodes of the network are already active.
Notice that attempting to activate a node without any active neighbors (or with sum of influences from these neighbors equal to zero) will result in the value of the formula equal to $\infty$.

Pseudocode of the dynamic programming algorithm is presented as Algorithm~\ref{alg:dynamic-programming}, while Figure~\ref{fig:dynpr-small} presents the intuition behind the algorithm by showing how it works on a specific graph structure.

\begin{algorithm}[t]
\setstretch{1.35}
\LinesNumbered
\DontPrintSemicolon
\SetAlgoNoEnd
\SetAlgoNoLine
\KwIn{A weighted network $(V,E,W)$, a seed node $\vs \in V$, and the number of nodes to activate $\goal$.}
\KwOut{Sequence of activation of $\goal$ nodes starting with $\vs$ with minimal expected time of activation.}
\For{$C \subseteq V$}{
	$\et^*[C] \gets \infty$\;
}
$\et^*[\{\vs\}] \gets 0$\;
$\seq^*[\{\vs\}] \gets \langle \vs\rangle$\;
\For{$k=1,\ldots,\goal-1$}{
	\For{$C \subset V : (|C|=k) \land (\et^*[C] < \infty)$}{
		\For{$i \in V : (i \notin C) \land (N(i) \cap C \neq \emptyset)$}{
			$\Delta \et \gets \frac{w_i}{\sum_{j \in N(i) \cap C}w_{ji}}$\;
			\If{$\et^*[C]+\Delta \et < \et^*[C\cup\{i\}]$}{
				$\et^*[C\cup\{i\}] \gets \et^*[C]+ \Delta \et$\;
				$\seq^*[C\cup\{i\}] \gets \seq^*[C] \oplus \langle i\rangle$\;
			}
		}
	}
}
\Return $\seq^*[\argmin_{C \subseteq V : |C|=\goal}\et^*[C]]$\;
\caption{Dynamic programming algorithm for strategic diffusion}
\label{alg:dynamic-programming}
\end{algorithm}

\begin{figure}[tbh]
\centering
\includegraphics[width=\linewidth]{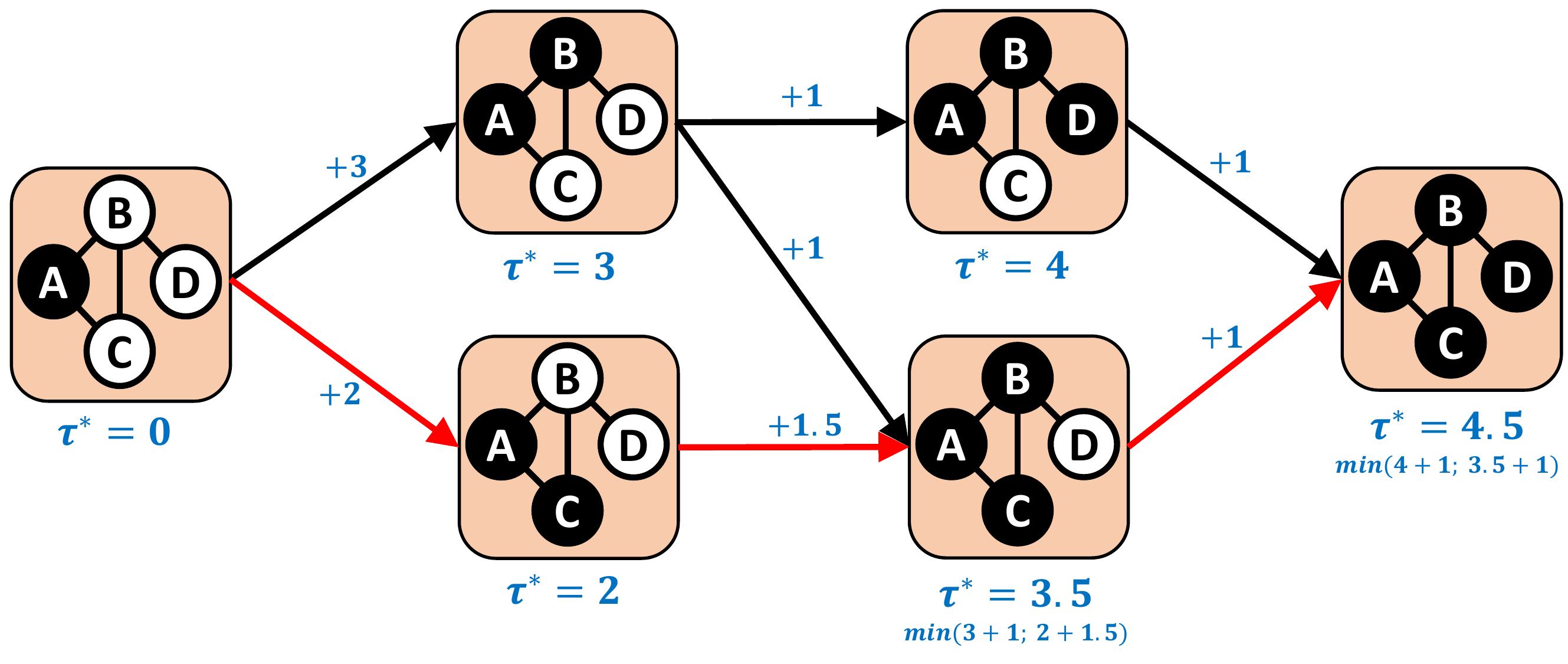}
\caption{Example of using dynamic programming.
Each large orange state represents a possible state of activation of the network, with black nodes representing active nodes, and white nodes representing inactive nodes.
The value of $\et^*$ indicates the minimal expected time required to reach this state of activation.
Value on each arrow represents the cost of activation of a single node, required to move between states.
Red arrows represent optimal sequence of activation, which corresponds to the least expensive path from the initial state to the state where entire network is activated.}
\label{fig:dynpr-small}
\end{figure}

In entry $\et^*[C]$ we compute the minimal expected time necessary to activate all nodes in set $C$, while in entry $\seq^*[C]$ we keep a sequence of activation allowing us to achieve this optimal expected time.
In the $k$-th execution of the loop in line~5 we compute values of entries in $\et^*$ and $\seq^*$ for sets $C$ such that $|C|=k+1$.
We do it by iterating (in loop in line~6) over all sets of nodes of size $k$ that can be activated when we start the process from the seed node $\vs$ and then iterate (in loop in line~7) over all possible nodes that can be targeted, \ie, nodes with non-zero probability of activation, when the set of active nodes is $C$.
In lines~8-11 we update the best way of activating nodes in set $C \cup \{ i \}$ when activating nodes in $C$ first and activating node $i$ afterwards has lower expected time than currently known fastest way to activate nodes in $C \cup \{ i \}$.
%In lines~12-13 we return the current best solution (sequence with $k$ nodes with lowest expected time of activation) if there are no sets of $k+1$ nodes that can be activated within time limit.

As for the implementation details, $\et^*$ and $\seq^*$ can be implemented as hash tables.
In this case checking whether $\et^*[C] < \infty$ is equivalent to checking whether $\et^*$ contains the entry for $C$ and therefore lines~1-2 can be omitted.

Notice also that during the $k$-th execution of the loop in line~5 we only need entries in tables $\et^*$ and $\seq^*$ for sets $C$ such that $|C|=k$.
Hence, to reduce memory requirements, we can only keep in memory entries from tables $\et^*$ and $\seq^*$ for two sizes of sets: the ones for size $k$, computed in the previous execution of the loop (or initialized in lines~3-4 in case of the first execution) and the ones for size $k+1$, being computed in the current execution of the loop.

Despite these optimization possibilities, the dynamic programming algorithm remains exponential, as it considers all subsets of nodes possible to be activated.
If the number of nodes to be activated is constant then the algorithm is polynomial.
\begin{theorem}
Algorithm~\ref{alg:dynamic-programming} solves the Optimal Partial Diffusion Sequence problem in time $\bigo(|V|^{z+1}|E|)$.
\end{theorem}

In Section~\ref{sec:treewidth}, we give a dynamic program that works more efficiently when the network has both bounded degree and bounded treewidth. Note that it is common to use dynamic programming for solving optimization problems on graphs with bounded treewidth, see, \eg, \cite{Bodlaender88}. 

An alternative approach to developing an algorithm of finding an effective way of activating the network would be to use the reinforcement learning techniques~\cite{kaelbling1996reinforcement}.
Reinforcement learning is based on the idea of making a sequence of decisions by finding a balance between exploitation (performing the currently best action to obtain short-term profits) and exploration (investigating other actions in the hope of a long-term gain).
In case of strategic diffusion, the exploitation would be activating a node with the highest activation probability, while the exploration would be activating a node with lower activation probability in order to make its neighbors easier to activate.
Nevertheless, an algorithm based on reinforcement learning would not guarantee that identified activation sequence is the optimal solution, unlike the presented dynamic programming algorithm, which offers such certainty.
Hence, we leave the development of an algorithm based on reinforcement learning as a potential future work.

We now discuss other ways of optimization for more restricted network structures.

\subsection{Decomposition into Biconnected Components}

We will now show that looking for the optimal way of activating all nodes of the network can be made simpler by using decomposition into biconnected components.

Cut node in a network (also called an articulation point) is a node the removal of which causes network to fall into two or more connected components.
By separating the network in cut nodes we obtain a decomposition into biconnected components.
Biconnected component is a maximal sub network (in terms of inclusion) such that removing any node from it will not disconnect it.
An example of a decomposition into biconnected components is shown in Figure~\ref{fig:decomposition}.
Notice that a copy of a cut node appears in every biconnected component it belongs to.

\begin{figure}[tbh]
\centering
\includegraphics[width=.8\linewidth]{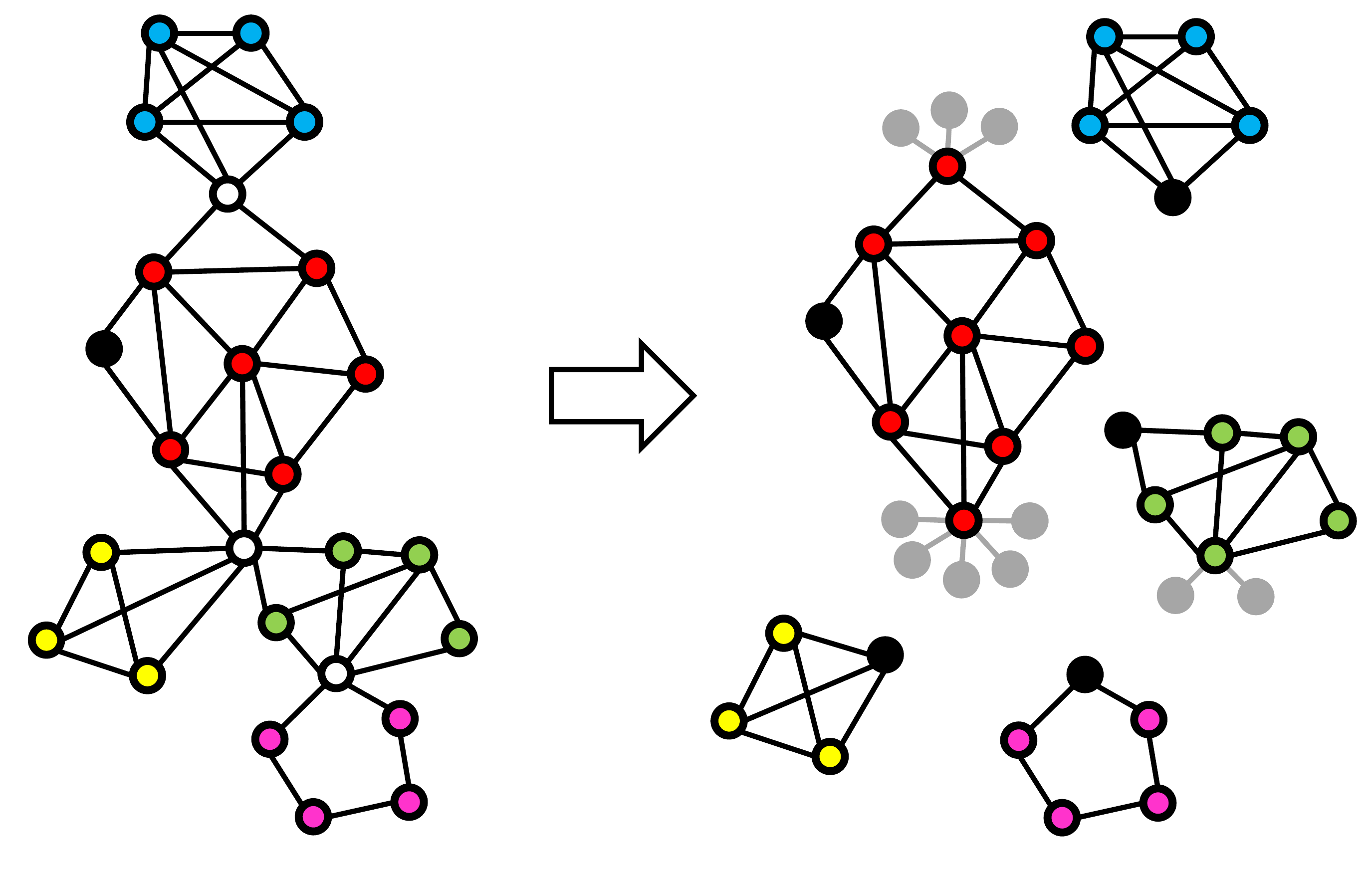}
\caption{Decomposition of graph into biconnected components.
Each biconnected component is marked with different color, with cut nodes marked white and seed nodes marked black.
Gray nodes mark stubs added to provide proper influence sum for original cut nodes.}
\label{fig:decomposition}
\end{figure}

The strategic diffusion process in every biconnected component can be considered separately.
This is because every biconnected components has only one possible starting point of the diffusion (being it either the seed node in case of biconnected component containing it or the cut node closest to the seed node in case of all other components).
This is indicated by black color of a node in Figure~\ref{fig:decomposition}.
Once this starting point of a given component $C$ is activated, the activation state in other biconnected components of the network does not affect activation  is $C$.

This is because activation probability of a node is only affected by the state of activation of its neighbors and the value of $w_i$ (to remind the reader, $w_i=\sum_{j \in N(i)}w_{ji}$).
Only cut nodes have neighbors in other biconnected components, hence for all non-cut nodes this observation is trivial.
As for any cut node $i$, notice that its neighbors in other biconnected components can only be activated after $i$ is activated (as all the paths between them and the seed node run through the cut node).
Therefore the only way in which neighbors of a cut node $i$ in other biconnected components affect the activation probability of $i$ is adding to the value of $w_i$.
To account for this fact, we create additional stubs connected to copies of $i$ that are not starting points of activation during the biconnected component decomposition (marked gray on Figure~\ref{fig:decomposition}).

To remind the reader, in this work we consider undirected networks (notice that the definition of biconnected components is different for directed networks).
However, we do not assume that the relation of influence is symmetric, \ie, it is possible that for $ij \in E$ we have $w_{ij} \neq w_{ji}$.
Since we consider the decomposition into disconnected components in the context of a given seed node, there is never a disambiguation in terms of which edge weights have to be used in case of the cut nodes.
To compute the activation probability of a cut node $i$, for each neighbor $j$ only weight $w_{ji}$ is taken into consideration.
What is more, only neighbors of $i$ that are closer to the seed node than $i$ can have a positive contribution into its activation probability, as the neighbors of $i$ that are further away from the seed node can only be activated after $i$ is activated.
At the same time, $i$ affects the activation time of its neighbor $j$ through the weight $w_{ij}$, and it can have a positive contribution both when $j$ is closer and when its further away to the source node than $i$.

\subsection{Optimal Solution for Networks with Bounded Treewidth}
\label{sec:treewidth}

We will now show a polynomial algorithm that finds an optimal way to activate a network with both treewidth and degree bounded by constants.

Let $G=(V,E,W)$ be an arbitrary network.
Let $(T,F)$ be a network where every node contains a subset of nodes from $V$, \ie, $\forall_{t \in T} t \subseteq V$ (we will call each such subset a \textit{bag}).
Such $(T,F)$ is a tree decomposition of $G$ (see ,\eg, \cite{RS84}) if and only if the following conditions are met:
\begin{itemize}
\item every node from $V$ is contained in at least one bag, \ie, $\bigcup_{t \in T}t = V$.
\item for every edge $e \in E$ there exists a bag that contains both ends of $e$, \ie, $\forall_{ij \in E} \exists_{t \in T} \{i,j\} \subseteq T$.
\item for every node $v \in V$ subnetwork of $(T,F)$ induced by bags containing $v$ is connected.
\end{itemize}

The treewidth $\tw$ of a decomposition is the size of its largest bag minus one, \ie, $\tw = \max_{t \in T}|t|-1$.
The treewidth of a network is the minimum over treewidths of all its tree decompositions. A problem is said to be {\it fixed-parameter tractable} \cite{r1999parameterized}, with respect to parameter $k,$ if any instance of the problem of size $N$ can be solved in time $f(k)\cdot N^{O(1)}$, for some computable function $f(\cdot)$.
We show that the Optimal Strategic Diffusion problem is fixed-parameter tractable with respect to the sum of treewidth and maximum degree.

\subsubsection{Full Diffusion}

To better explain the idea, we first give an algorithm for the full diffusion case.

\begin{theorem}
Let $G=(V,E,W)$ be a network with maximal degree $\dm$.
There exists an algorithm that, given a tree decomposition of treewidth $\tw$, finds an optimal way of activating nodes in $V$ in time $\bigo(\tw\dm(\tw\dm)!^2 n)$.
\end{theorem}

In what follows, we assume that the tree decomposition is rooted in some node $t_R$, the same for bottom-up and top-down order.
We use the following notation:
\begin{itemize}
\item $c(t)$ denotes the sequence of children of $t$;
\item $\seq|_{X}$ denotes the subsequence of $\seq$ consisting only of the nodes in set $X$;
\item $\seq_{\triangleleft x}$ denotes the subsequence of $\seq$ consisting of all elements preceding $x$;
\item $\seq_{\blacktriangleright x}$ denotes the subsequence of $\seq$ consisting of $x$ as well as all elements following $x$;
\item $T_{topdown}$ denotes sequence of nodes in $T$ in a top-down order.
\end{itemize}

In what follows we will call two permutations $\seq$ and $\seq'$ of nodes from $G$ \textit{compatible}, denoted $\seq\sim\seq'$,  if and only if they activate the nodes that they have in common in the same order, \ie, $\seq|_{\seq'}=\seq'|_{\seq}$.

For each node of $t$ the tree decomposition $T$ we compute a record consisting of the following fields:
\begin{itemize}
\item $\Upsilon[t]$ stores all permutations of nodes in bag $t$ and their neighbors that can be a part of a valid solution to the problem;
\item $\et[t,\seq,i]$ stores the expected time of activation of every node in $i \in t$ when activated in the order given by $\seq$;
\item $\et^*[t,\seq]$ stores the smallest expected time to activate all nodes in the subnetwork of $G$ induced by the nodes in the subtree of $T$ rooted at $t$, when the nodes in $t$ and their neighbors are activated in the order given by $\seq$.
\end{itemize}

To compute the records we traverse the tree decomposition in a bottom-up order (given the root $t_R$) and for every node $t$ we perform the following steps:
\begin{enumerate}
\item \label{proc1-step1}
Compute:
$$
\Upsilon[t] = \bigg\{ \seq \in \Seq^*(t \cup \bigcup_{i \in t}N(i)) :
	\left(\vs \notin t \lor \seq_1 = \vs\right)
	\land \left(\forall_{\substack{\seq_i \in t :\\ \seq_i \neq \vs}} \exists_{j < i} \seq_j \in N(\seq_i)\right)
	\land \left(\forall_{t' \in c(t)}\exists_{\seq' \in \Upsilon[t']} \seq'\sim \seq\right)
\bigg\}.
$$

\item \label{proc1-step2}
For every $\seq \in \Upsilon[t]$ and every $\seq_i \in \seq|_t$ compute:
$$
\et[t,\seq,\seq_i] = \frac{w_{\seq_i}}{\sum_{j < i}w_{\seq_j\seq_i}}.
$$

\item \label{proc1-step3}
For every $\seq \in \Upsilon[t]$ compute:
$$
\et^*[t,\seq] =
\sum_{\seq_i \in \seq|t}\et[t,\seq,\seq_i]
+ \sum_{t' \in c(t)} \min_{\seq' \in \Upsilon[t'] : \seq'\sim \seq}
\left(\et^*[t',\seq'] - \sum_{\mathclap{\seq'_i \in \seq'|_t}} \et[t',\seq',\seq'_i]\right).
$$
\end{enumerate}

After traversing the tree in this fashion, the minimal time required to activated entire network $G$ starting with $\vs$ can be identified as $\min_{\seq \in \Upsilon[t_R]} \et^*[t_R,\seq]$.
In order to reconstruct the sequence allowing to activate entire network in the optimal time, we can run Algorithm~\ref{alg:reconstruct-full}.

\begin{algorithm}[t]
\LinesNumbered
\DontPrintSemicolon
\SetAlgoNoEnd
\SetAlgoNoLine
\KwIn{Tree decomposition $(T,F)$ of a weighted network $(V,E,W)$, with computed values of $\et^*$.}
\KwOut{Sequence of activation of nodes in $V$ starting with $\vs$ with minimal expected time of activation.}
$\seq^* \gets \langle \rangle$\;\label{alg1-ln12}
\For{$t \in T_{topdown}$}{\label{alg1-ln13}
	$\seq \gets \argmin_{\seq' \in \Upsilon[t] : \seq'\sim\seq^*}\et^*[t,\seq']$\;\label{alg1-ln14}
	$\seq' \gets \langle \rangle$\;\label{alg1-ln15}
	\For{$\seq_i \in \seq$}{\label{alg1-ln16}
		\If{$\seq_i \notin \seq^*$}{\label{alg1-ln17}
			$\seq' \gets \seq' \oplus \langle \seq_i \rangle$\;\label{alg1-ln18}
		} \Else {\label{alg1-ln19}
			$\seq^* \gets \seq^*_{\triangleleft\seq^*_i} \oplus \seq' \oplus \seq^*_{\blacktriangleright\seq^*_i}$\;\label{alg1-ln20}
			$\seq' \gets \langle \rangle$\;\label{alg1-ln21}
		}
	}
	$\seq^* \gets \seq^* \oplus \seq'$\;\label{alg1-ln22}
}
\Return $\seq^*$\;\label{alg1-ln23}
\caption{Algorithm reconstructing the optimal way of activating all nodes of a network with bounded treewidth and degree.}
\label{alg:reconstruct-full}
\end{algorithm}

Let us now comment on the procedure of filling the records.

In step~\ref{proc1-step1} we gather all permutations of nodes in bag $t$ and their neighbors that can be a part of a valid solution to the problem, according to three conditions.
The first condition asserts that if the permutation contains the seed node, it is activated as the very first node.
The second condition assures that every node in bag $t$ other than the source node has at least one active neighbor at the moment of activation (notice that for every node in $t$ all of its neighbors are present in the permutation).
The third condition provides that the permutation $\seq$ allows to activate all nodes in the subtree of $t$, \ie, that for every child of $t$ in the tree decomposition there exists at least one valid permutation $\seq'$ that activates nodes from $\seq$ in the same order as $\seq$.

In step~\ref{proc1-step2} we simply compute the time of activation of every node in bag $t$, according to the definition given in Section~\ref{sec:problem-definition}, and store it in table $\et$.
Again, notice that for every node in $t$ all of its neighbors are present in the permutation, hence we have enough information to compute its expected time of activation.

Finally, in step~\ref{proc1-step3} we compute the optimal time of activation of all nodes in the subtree of $t$ while using permutation $\seq$ and store it in table $\et^*$.
The expression consists of a sum of time of activation of the nodes in bag $t$ and a sum over all children of $t$, where for each of them we compute the minimum over all compatible permutations and subtract the time of activation of nodes in $t$.
Notice that since $(T,F)$ is the tree decomposition, the only possible overlap between nodes in the subtrees of different children of $t$ are nodes in bag $t$.
Otherwise the graph induced by bags containing such overlapping nodes would not be connected, and hence one of the conditions of being the tree decomposition would not be met for $(T,F)$.

After having filled tables $\Upsilon$, $\et^*$ and $\et$, we traverse the tree decomposition again, this time in a top-down order using Algorithm~\ref{alg:reconstruct-full}.
We construct an optimal solution to the problem on variable $\seq^*$.
In line~\ref{alg1-ln14} we select as $\seq$ the permutation with minimal time necessary to activate all nodes in the subtree of $t$, that is compatible with the solution constructed so far.
In lines~\ref{alg1-ln15}-\ref{alg1-ln22} we merge $\seq$ into sequence $\seq^*$ using auxiliary variable $\seq'$.

As for the time complexity of the algorithm, the most costly operations (both of which are equally expensive) are validating existence of a compatible sequence for every child of $t$ in the third condition in step~\ref{proc1-step1} and computing the time necessary to activate the nodes in the subtrees of all children of $t$ in the second sum in step~\ref{proc1-step3}.
For assessing the time complexity of the algorithm we will focus on the cost in step~\ref{proc1-step3}.
Since every node $t'$ is a child of at most one other node in the tree decomposition, for every $t' \in T$ the computation of minimum is executed $\bigo((\tw\dm)!)$ times (the number of different permutations $\seq \in \Upsilon[t]$).
Since there are $\bigo(n)$ nodes in the tree decomposition, the computation of minimum is executed $\bigo((\tw\dm)!n)$ times.
The cost of executing it once is $\bigo(\tw\dm(\tw\dm)!)$, since we have to check $\bigo((\tw\dm)!)$ many permutations in $\Upsilon[t']$ and for each of them check whether $\seq'|_{\seq}=\seq|_{\seq'}$ in time $\bigo(\tw\dm)$.
Hence, the total time complexity of the algorithm is $\bigo(\tw\dm(\tw\dm)!^2 n)$.

\subsubsection{Partial Diffusion}

We next present an algorithm for partial diffusion in networks with bounded treewidth and bounded degree.
\begin{theorem}
Let $G=(V,E,W)$ be a network with maximal degree $\dm$.
There exists an algorithm that, given a tree decomposition of treewidth $\tw$, finds an optimal way of activating $z$ nodes in $V$ in time $\bigo(\tw\dm(\tw\dm)!^2 z^2n)$.
\end{theorem}

We use essentially the same notation as in previous section.
However, in what follows we will call two permutations $\seq\in\Seq(X)$ and $\seq'\in\Seq(X')$ of nodes from $G$ \textit{compatible}, denoted $\seq\sim\seq'$, if and only if they activate the nodes that they have in common in the same order, \ie, $\seq|_{\seq'}=\seq'|_{\seq}$, and agree on the non-activated nodes, \ie, $\seq|_{X'\setminus\seq'}=\seq'|_{X\setminus\seq}=\langle \rangle$ (this difference is the result of considering also permutations that are not full, \ie, permutations $\seq\in\Seq(X)$ such that $\seq < |X|$).

The record that we compute for each node of $t$ the tree decomposition $T$ consisting now of the following fields:
\begin{itemize}
\item $\Upsilon[t]$ stores all permutations of nodes in bag $t$ and their neighbors that can be a part of a valid solution to the problem;
\item $\et[t,\seq,i]$ stores the expected time of activation of every node in $i \in t$ when activated in the order given by $\seq$;
\item  $\et^*[t,\seq,k]$ denotes the smallest expected time to activate $k$ nodes in the subnetwork of $G$ induced by the nodes in the subtree of $T$ rooted at $t$, when the nodes in $t$ and their neighbors are activated in the order given by $\seq$.
\end{itemize}
Hence, the only difference in comparison to the full diffusion version is that table $\et^*$ is additionally indexed with the number of nodes to activate.

To compute the records we traverse the tree decomposition in a bottom-up order (given the root $t_R$) and for every node $t$ we perform the following steps:
\begin{enumerate}
\item \label{proc2-step1}
Compute:
$$
\Upsilon[t] = \bigg\{ \seq \in \Seq(t \cup \bigcup_{i \in t}N(i)) :
	\left(\vs \notin t \lor \seq_1 = \vs\right)
	\land \left(\forall_{\substack{\seq_i \in t :\\ \seq_i \neq \vs}} \exists_{j < i} \seq_j \in N(\seq_i)\right) 
	\land \left(\forall_{t' \in c(t)}\exists_{\seq' \in \Upsilon[t']} \seq'\sim \seq\right)
\bigg\}.
$$

\item \label{proc2-step2}
For every $\seq \in \Upsilon[t]$ and every $\seq_i \in \seq|_t$ compute:
$$
\et[t,\seq,\seq_i] = \frac{w_{\seq_i}}{\sum_{j < i}w_{\seq_j\seq_i}}.
$$

\item \label{proc2-step3}
For every $\seq \in \Upsilon[t]$ and for $k=|\seq|_t|,\ldots,z$ compute the value of $\et^*[t,\seq,k]$ using Algorithm~\ref{alg:treewidth-partial-sub}.
\end{enumerate}

\begin{algorithm}[tbhp]
\LinesNumbered
\DontPrintSemicolon
\SetAlgoNoEnd
\SetAlgoNoLine
\KwIn{Tree decomposition $(T,F)$ of a weighted network $(V,E,W)$, with values of $\et^*$ computed so far, node $t$ with sequence of children $c(t)=\langle t_1,\ldots,t_{|c(t)|} \rangle$.}
\KwOut{The value of $\et^*[t,\seq,k]$.}
\For{$t_i \in \langle t_1,\ldots,t_{|c(t)|} \rangle$}{
	\For{$m=0,\ldots,k-\left|\seq|_t\right|$}{
		\If{$i=1$}{
			$\et'[t,1,\seq,m]\gets\min\limits_{\substack{\seq' \in \Upsilon[t_1]:\\ \seq'\sim\seq}}\left(\et^*[t_1,\seq',m+|\seq'|_{t}|] - \sum_{\seq'_i \in \seq'|_{t}}\et[t_1,\seq',\seq'_i]\right)$\;
		} \Else {
			$\et'[t,i,\seq,m] \gets \min\limits_{m'\in\{0,\ldots,m\}} \bigg( \et'[t,i-1,\seq,m']$\; \label{alg2-line6}
			\nonl \hspace{4cm} $+\min\limits_{\substack{\seq' \in \Upsilon[t_i] :\\ \seq'\sim\seq}}\left(\et^*[t_i,\seq',m-m'+|\seq'|_{t}|] - \sum_{\seq'_i \in \seq'|_{t}}\et[t_i,\seq',\seq'_i]\right) \bigg)$\;
		}
	}
}
\Return $\sum_{\seq_i \in \seq|_{t}} \et[t,\seq,\seq_i] + \et'\left[t,|c(t)|,\seq,k-\left|\seq|_t\right|\right]$\;
\caption{Algorithm computing the value of $\et^*[t,\seq,k]$.}
\label{alg:treewidth-partial-sub}
\end{algorithm}

After traversing the tree in this fashion, the minimal time required to activate $z$ nodes in network $G$ starting with $\vs$ can be identified as $\min_{\seq \in \Upsilon[t_R]} \et^*[t_R,\seq,z]$.
In order to reconstruct the sequence allowing to activate $z$ nodes in the optimal time, we can run Algorithm~\ref{alg:reconstruct-partial}.

\begin{algorithm}[t]
\LinesNumbered
\DontPrintSemicolon
\SetAlgoNoEnd
\SetAlgoNoLine
\KwIn{Tree decomposition $(T,F)$ of a weighted network $(V,E,W)$, with computed values of $\et^*$ and $\et'$, a bag $t$, a sequence $\seq^* \in \Seq(V)$ and an integer $k \leq n$.}
\KwOut{Sequence of activation of $z$ nodes in $V$ starting with $\vs$ with minimal expected time of activation.}
\If{$k>0$}{
	$\seq \gets \argmin_{\seq' \in \Upsilon[t] : \seq' \sim\seq^*}\et^*[t,\seq',k]$\;
	$\seq' \gets \langle \rangle$\;
	\For{$\seq_i \in \seq$}{
		\If{$\seq_i \notin \seq^*$}{	
			$\seq' \gets \seq' \oplus \langle \seq_i \rangle$\;
		} \Else {
			$\seq^* \gets \seq^*_{\triangleleft\seq^*_i} \oplus \seq' \oplus \seq^*_{\blacktriangleright\seq^*_i}$\;
			$\seq' \gets \langle \rangle$\;
		}
	}
	$\seq^* \gets \seq^* \oplus \seq'$\;
	$k'\gets k-|\seq|_t|$\;
	\For{$t_i \in \langle t_{|c(t)|},\ldots,t_1 \rangle$}{ 
		$m \gets \argmin_{m'\in\{0,\ldots,k'\}} \bigg( \et'[t,i-1,\seq,m']$\;
 		\nonl \hspace{4cm} $+\min_{\substack{\seq' \in \Upsilon[t_i] :\\ \seq'\sim\seq}}\left(\et^*[t_i,\seq',k'-m'+|\seq'|_{t}|] - \sum_{\seq'_i \in \seq'|_{t}}\et[t_i,\seq',\seq'_i]\right) \bigg)$\;
		$\seq^*\gets\text{ReconstructPartial}(t_i,\seq^*,k'-m+|\seq'|_{t}|)$\;
		$k'\gets m$\;
    }
}
\Return $\seq^*$\;
\caption{ReconstructPartial$(t,\seq^*,k)$: algorithm reconstructing the optimal way of activating $k$ nodes in a network with bounded treewidth and degree.}
\label{alg:reconstruct-partial}
\end{algorithm}

Let us now comment on the procedure of filling the records.

Steps~\ref{proc2-step1} and~\ref{proc2-step2} are almost identical to those of the algorithm for full diffusion, with notable difference being that this time we consider sequences that are not necessarily full.

In step~\ref{proc2-step3} we call Algorithm~\ref{alg:treewidth-partial-sub} to compute the value of $\et^*[t,\seq,k]$, \ie, the optimal time of activating $k$ nodes in the subtree of $t$ while using permutation $\seq$.
To this end, we visit all the children of $t$, in the order $t_1,\ldots, t_{|c(t)|}$.
For $m\in\{0,\ldots,k-|\seq|_t|\}$ we compute $\et'[t,i,\seq,m]$, the minimum expected time required to activate $m$ nodes \textit{other than the ones activated in $t$}, in the subnetwork induced by union of the subtrees rooted at $t_1,\ldots,t_i$.
We do this using the values of $\et'$ computed for children of $t$ preceding $t_i$, \ie, for $t_1,\ldots,t_{i-1}$, as well as the values of $\et^*$ computed for $t_i$.
In the selection process $m'$ denotes the number of nodes activated in the subnetwork induced by union of $t_1,\ldots,t_{i-1}$, while the rest of the $m$ nodes is activated in the subtree rooted in $t_i$.
The returned value of the minimal time of activation of $k$ nodes in the subtree of $t$ while using permutation $\seq$ is computed as the sum of the time of activation of nodes from $t$ and the optimal time of activating the remaining $k-\left|\seq|_t\right|$ nodes in all $|c(k)|$ children of $t$.

The time complexity of the algorithm is similar to the full diffusion case, except that we have now the two additional loops, one over $k$ in step~\ref{proc2-step3} and the other over $m$ in line~\ref{alg2-line6} of Algorithm~\ref{alg:treewidth-partial-sub}, which contribute an additional factor of $\bigo(z^2)$ to the running time.

We are unable to compute an optimal solution in polynomial time in general case.
We may however hope to find a polynomial approximation algorithm.
We now move to the analysis of possible ways of approximating the optimal solution.

%\section{Hardness of Approximation}

We now move to assessing the possibility of approximating the optimal solution to the Optimal Diffusion Sequence Problem.

\label{sec:approximation}
\subsection{Lower Bounds on Heuristic Algorithms}

Alshamsi~\etal~\cite{alshamsi2018optimal} suggest two heuristic strategies for activating a network in a process of strategic diffusion: the greedy strategy (always target node with the highest probability of activation) and the majority strategy (always target node with the highest number of active neighbors).

We now show that neither of these strategies approximates an optimal solution to within a constant ratio.

\begin{theorem}
\label{thrm:greedy}
There exists no constant $r > 1$ such that choosing the sequence of activation using the greedy strategy is an $r$-approximation algorithm. In particular, the approximation ratio of the greedy strategy is $\Omega(\log n)$.
\end{theorem}

\begin{proof}
We will now show a series of networks where the approximation ratio of the greedy algorithm goes to infinity.

For a given $k \in \N$ we construct a network $G(k)$, where the set of nodes is $V=\{\vs,a_1,\ldots,a_{k^2},b_1,\ldots,b_{k-1}\}$ and where the set of edges is:
$$
E=\bigcup_{i=1}^{k^2}\{\vs a_i\} \cup \bigcup_{i=1}^{k^2}\bigcup_{j=1}^{k-1}\{a_i b_j\}.
$$
We assume that the weight of every edge is $1$.
The structure of network $G(k)$ is presented in Figure~\ref{fig:greedy-majority}.

\begin{figure}[tbh]
\centering
\includegraphics[width=.5\linewidth]{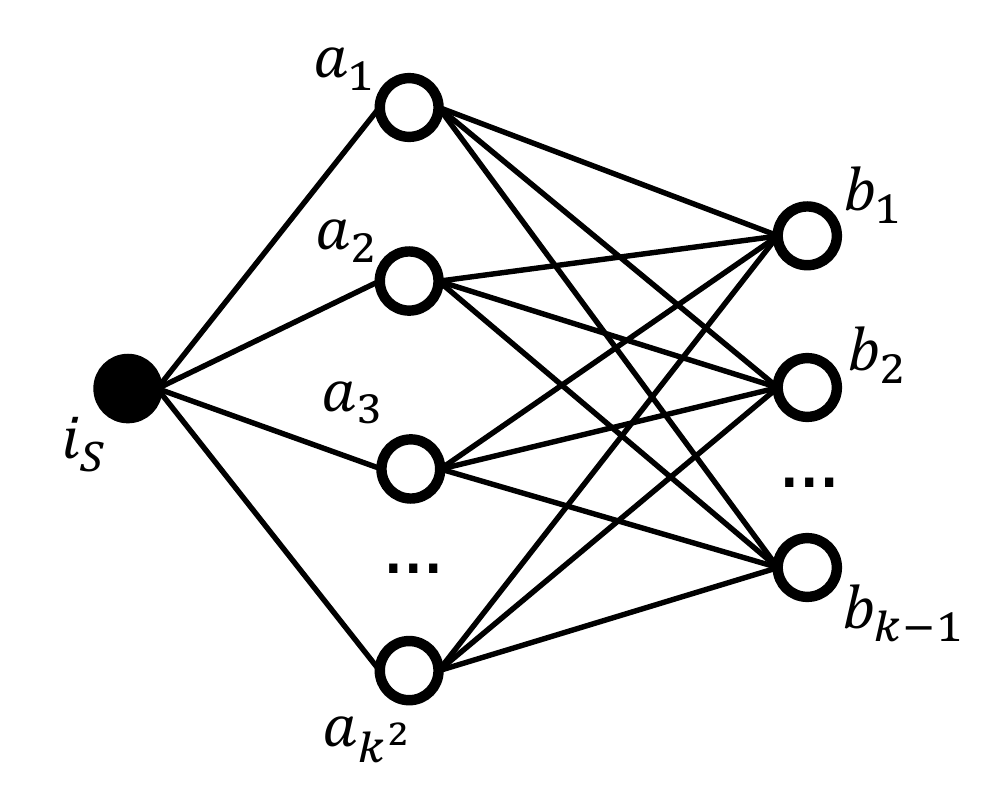}
\caption{Network $G(k)$ constructed for the proofs of Theorems~\ref{thrm:greedy} and~\ref{thrm:majority}.}
\label{fig:greedy-majority}
\end{figure}

Let $x$ be the number of currently activated $a_i$ nodes and let $y$ be the number of currently activated $b_i$ nodes.
We have that $\et(a_i) = \frac{k}{y+1}$ and we have that $\et(b_i) = \frac{k^2}{x}$.

We will now analyze two different ways of activating all nodes in network $G(k)$.
First, let us consider the greedy algorithm.
Let the indices of nodes $a_i$, as well as the indices of nodes $b_j$, be ordered according to the order of activation, \ie, node $a_1$ is the first activated node $a_i$, while node $b_1$ is the first activated node $b_i$.
Notice that at least $ik$ nodes $a$ has to be activated before the activation of node $b_{i}$.
We will prove this claim by contradiction.
Assume that $b_{i}$ can have $j<ik$ active neighbors at the moment of activation.
Take node $a_{j+1}$ (the first node $a$ activated after $b_{i}$).
At the moment of activation of $b_{i}$ its probability of activation is $\frac{j}{k^2}$, while the probability of activation of $a_{j+1}$ is $\frac{i}{k}$.
Since we are using greedy algorithm and $b_{i}$ was activated before $a_{j+1}$, we have to have $\frac{j}{k^2} \geq \frac{i}{k}$.
However, this is not true since $j<ik$.
We have a contradiction, therefore at least $ik$ nodes $a$ has to be activated before the activation of node $b_{i}$.

Because of this at least $k$ nodes $a_i$ has to be activated with only one neighbor active, then at least $k$ nodes $a_i$ has to be activated with only two neighbors active, and so on.
Focusing only on the time of activation of nodes $a_i$ we have that the total expected time of activation for the greedy algorithm is:
$$
\et_G(G(k)) \geq \sum_{i=1}^{k^2} \et(a_i) \geq \sum_{j=1}^{k} k \frac{k}{j}  = k^2 H_k
$$
where $H_k$ is the $k$-th harmonic number.

Now, let us consider an algorithm $A$, where we first activate $k$ nodes $a_i$, then all nodes $b_i$ and finally the remaining nodes $a_i$.
Time of activation of the entire network is then:
$$
\et_A(G(k)) = k^2+(k-1)k+(k^2-k) = 3k^2-2k
$$
since each of the first $k$ nodes $a_i$ is activated in expected time $k$ (as it only has one active neighbor, namely $\vs$), each node $b_i$ is activated in expected time $k$ (as it has degree $k^2$ and $k$ active neighbors at the moment of activation) and each of the remaining $k^2-k$ nodes $a_i$ is activated in expected time $1$.

Consider sequence of networks $G(k)$.
For this sequence we have
$$
\lim_{k \to \infty} \frac{\et_G(G(k))}{\opt(G(k))} \geq \lim_{k \to \infty} \frac{\et_G(G(k))}{\et_A(G(k))} \geq \lim_{k \to \infty} \frac{H_k}{3} = \infty.
$$
Therefore, solution provided by the greedy algorithm can be arbitrarily worse than the optimal solution.
\end{proof}

We also show analogical result for the majority strategy.

\begin{theorem}
\label{thrm:majority}
There exists no constant $r > 1$ such that choosing the sequence of activation using the majority strategy is an $r$-approximation algorithm. In particular, the approximation ratio of the majority strategy is $\Omega(\log n)$.
\end{theorem}

\begin{proof}
We will now show that for the series of networks constructed in the proof of Theorem~\ref{thrm:greedy} the approximation ratio of the majority algorithm also goes to infinity.

Let $x$ be the number of currently activated $a_i$ nodes and let $y$ be the number of currently activated $b_i$ nodes.
We have that $\et(a_i) = \frac{k}{y+1}$ and we have that $\et(b_i) = \frac{k^2}{x}$.

Let us consider the expected activation time of an entire network $G(k)$ using majority algorithm.
Let the indices of nodes $a_i$, as well as the indices of nodes $b_j$, be ordered according to the order of activation, \ie, node $a_1$ is the first activated node $a_i$, while node $b_1$ is the first activated node $b_i$.
Notice that node $b_{i}$ at the moment of activation has at most $i+1$ active neighbors.
We will prove this claim by contradiction.
Assume that $b_{i}$ can have $j>i+1$ active neighbors at the moment of activation.
Consider number of active neighbors at the moment of activation of node $a_{j}$ (the last node $a$ activated before $b_{i}$).
Since we are using the majority algorithm, it has to have at least $j-1$ active neighbors (otherwise node $b_{i}$, having at the moment $j-1$ active neighbors, would have been chosen before node $a_{j}$).
However, node $a_{j}$ can have at most $i<j-1$ active neighbors, namely nodes $b_1, \ldots, b_{i-1}$ and node $\vs$ (as node $b_{i}$ is not active yet).
We have a contradiction, therefore $b_{i}$ at the moment of activation has at most $i+1$ active neighbors.

Focusing only on the time of activation of nodes $b_i$ we have that the total expected time of activation for the majority algorithm is:
$$
\et_C(G(k)) \geq \sum_{i=1}^{k-1} \et(b_i) \geq \sum_{i=1}^{k-1} \frac{k^2}{i+1} = k^2(H_k - 1)
$$
where $H_k$ is the $k$-th harmonic number.

Let $A$ be the alternative algorithm described in the proof of Theorem~\ref{thrm:greedy}.
Consider sequence of networks $G(k)$.
For this sequence we have
$$
\lim_{k \to \infty} \frac{\et_C(G(k))}{\opt(G(k))} \geq \lim_{k \to \infty} \frac{\et_C(G(k))}{\et_A(G(k))} \geq \lim_{k \to \infty} \frac{H_k}{3} = \infty
$$
Therefore, solution provided by the majority algorithm can be arbitrarily worse than the optimal solution.
\end{proof}

\subsection{Inapproximability}

Finally, we show that in fact the problem cannot be approximated within a ratio of $(1-\epsilon) \ln n$ for any $\epsilon > 0$, unless $P=NP$.

\begin{theorem}
\label{thrm:logn}
The Optimal Partial Diffusion Sequence problem cannot be approximated within a ratio of $(1-\epsilon) \ln n$ for any $\epsilon > 0$, unless $P=NP$.
\end{theorem}

\begin{proof}
In order to prove the theorem, we will use the result by Dinur and Steurer~\cite{dinur2014analytical} that the Minimum Set Cover problem cannot be approximated within a ratio of $(1-\epsilon) \ln n$ for any $\epsilon > 0$, unless $P=NP$.

Let $X=(U,S)$ be an instance of the Minimum Set Cover problem.
To remind the reader, $U$ is the universe $\{u_1, \ldots, u_{|U|}\}$, while $S$ is a collection $\{S_1, \ldots, S_{|S|}\}$ of subsets of $U$.
The goal here is to find subset $S^* \subseteq S$ such that the union of $S^*$ equals $U$ and the size of $S^*$ is minimal.

First, we will show a function $f(X)$ that based on an instance of the Minimum Set Cover problem $X$ constructs an instance of the Optimal Partial Diffusion Sequence problem.

Let network $G(X)$ be defined as follows (an example of such network is presented in Figure~\ref{fig:logn}):
\begin{itemize}
\item \textbf{The set of nodes:}
For every $S_i \in S$, we create three nodes, denoted by $S_i$, $q_i$ and $q'_i$.
For every $u_i \in U$, we create $|S|+1$ nodes, denoted by $u_{i,1}, \ldots, u_{i,|S|+1}$.
Additionally, we create a single node $\vs$.
\item \textbf{The set of edges:}
For every node $S_i$, we create an edge $S_i \vs$.
For every node $q_i$, we create edges $q_i S_i$ and $q_i q'_i$.
Finally, for every node $u_{i,i'}$ and every $S_j$ such that $u_i \in S_j$ we create an edge $u_{i,i'} S_j$.
\item \textbf{The weight matrix:}
For every edge $S_i q_i$ we set its weights to $w_{S_i q_i}=0$ and $w_{q_i S_i}=\alpha-1$, where $\alpha = z|S|^{\lambda+1}$ for some $\lambda > 0$.
For every edge $u_{i,i'} S_j$ we set its weights to $w_{u_{i,i'} S_j}=0$ and $w_{S_j u_{i,i'}}=1$.
For all other pairs of nodes connected with an edge we set their weights to $1$.
\end{itemize}

\begin{figure}[tbh]
\centering
\includegraphics[width=.6\linewidth]{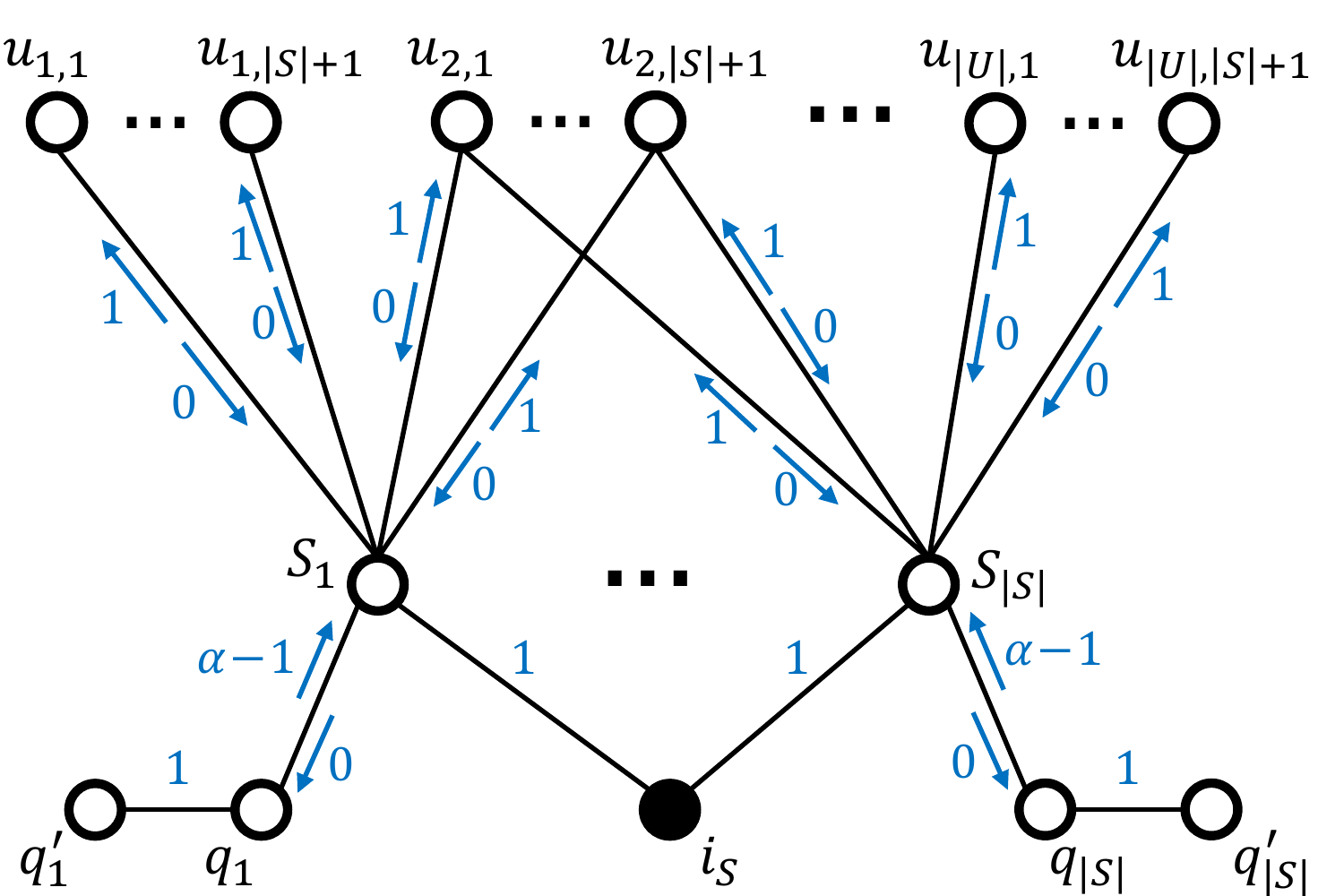}
\caption{Network $G(X,r)$ constructed for the proof of Theorem~\ref{thrm:logn}.
Blue numbers next to edges express their weights.
If there are no arrows next to edge $ij$ then $w_{ij}=w_{ji}$.
Otherwise weight $w_{ij}$ is denoted next to the arrow pointing towards node $j$, and weight $w_{ji}$ is denoted next to the arrow pointing towards node $i$.}
\label{fig:logn}
\end{figure}

To complete the constructed instance of the Optimal Partial Diffusion Sequence problem, we set the seed node to be $\vs$ and we set the number of nodes to be activated to $\goal=|U|(|S|+1)+1$.
Hence, the formula of function $f$ is $f(X)=(G(X),\vs,\goal)$.

Let $\seq$ be the solution to the constructed instance of the Optimal Partial Diffusion Sequence problem.
The function $g$ computing corresponding solution to the instance $X$ of the Minimum Set Cover problem is now $g(X,\seq^*)= S \cap \seq^*$, \ie, $S^*$ is the set of all sets $S_i$ such that their corresponding nodes $S_i$ appear in sequence $\seq^*$.

Now, we will show that $g(X,\seq^*,r)$ is indeed a correct solution to $X$, \ie, that it covers the entire universe.
Notice that none of the nodes $q_i$ nor $q'_i$ can get activated while $\vs$ is the seed node, since the influence of $S_i$ on $q_i$ is $0$.
Hence, sequence $\seq^*$ can consist only of $\vs$ and nodes in $S \cup U$.
Moreover, since we have to activate $|U|(|S|+1)$ of the nodes $S \cup U$, at least one node from each group $u_{i,1}, \ldots, u_{i,|S|+1}$ have to be activated.
The only way to do it is to activate a node $S_j$ such that $u_i \in S_j$.
Therefore, for every $u_i \in U$ there exists a node $S_j \in \seq^*$ such that $u_i \in S_j$ and $g(X,\seq)$ has to be a set cover of $U$.

We will now prove three lemmas concerning functions $f$ and $g$.

\begin{lemma}\label{lem:logn1}
Size of the solution to the given instance of the Minimum Set Cover problem returned by function $g$ is lesser or equal than the expected time of activation of the corresponding solution to the constructed instance of the Optimal Partial Diffusion Sequence problem $\seq$ divided by $\alpha$, \ie, $|g(X,\seq)| \leq \frac{\et(\seq)}{\alpha}$.
\end{lemma}
\begin{proof}
Let $\scv$ be the number of nodes $S_i$ in sequence $\seq$.
Since the expected time of activation of a node $S_i$ is always $\alpha$, we have $\et(\seq) \geq \scv \alpha$.
From the definition of $g$ we have $|g(X,\seq)|=\scv$, hence we have $|g(X,\seq)| \leq \frac{\et(\seq)}{\alpha}$.
\end{proof}

\begin{lemma}\label{lem:logn2}
For sequence $\seq$, the solution of $f(X)$, and the corresponding solution to $X$ returned by function $g$ we have $\et(\seq) \leq |g(X,\seq)| \alpha \left(1+\frac{1}{|S|^{\lambda}}\right)$.
\end{lemma}
\begin{proof}
Let $\scv$ be the number of nodes $S_i$ in sequence $\seq$.
Since the expected time of activation of a node $S_i$ is always $\alpha$ and the expected time of activation of a node $u_{i,j}$ is smaller or equal than $|S|$, we have $\et(\seq) \leq \scv \alpha + (\goal-\scv-1)|S| \leq \scv \alpha + \goal |S| \leq \scv \alpha + \scv \goal |S| = \scv \alpha + \scv \frac{\alpha}{|S|^{\lambda}}$.
From the definition of $g$ we have $|g(X,\seq)|=\scv$, hence we have $\et(\seq) \leq |g(X,\seq)| \alpha \left(1+\frac{1}{|S|^{\lambda}}\right)$.
\end{proof}

\begin{lemma}\label{lem:logn3}
An optimal solution to the constructed instance of the Optimal Partial Diffusion Sequence problem $\seq^*$ corresponds to an optimal solution to the given instance of the Minimum Set Cover problem $S^*$, \ie, $S^*=g(X,\seq^*)$.
\end{lemma}
\begin{proof}
As noted above, if $\seq$ is a solution to $f(X)$, then $g(X,\seq)$ is also a correct solution to $X$, \ie, $g(X,\seq)$ covers entire universe.
Since the expected time of activation of every node $S_i$ is $\alpha$, the expected time of activation of every node $u_{i,j}$ is smaller or equal than $|S|$, and $|S| < \alpha$, the optimal solution to $f(X)$ is the one that minimizes the number of nodes $S_i$ is $\seq$, hence, the one that minimizes $|g(X,\seq)|$.
\end{proof}

Now, assume that there exists an $r$-approximation algorithm for the Optimal Partial Diffusion Sequence problem where $r=(1-\epsilon)\ln n$ for some $\epsilon > 0$.
Let us use this algorithm to solve the constructed instance $f(X)$ and consider solution $g(X,\seq)$ to the given instance of the Minimum Set Cover problem.

We then have:
$$
|g(X,\seq)| \leq \frac{\et(\seq)}{\alpha} \leq \frac{r\et(\seq^*)}{\alpha} \leq r \left(1+\frac{1}{|S|^{\lambda}}\right)|g(X,\seq^*)| = r \left(1+\frac{1}{|S|^{\lambda}}\right)|S^*|,
$$
where first inequality comes from Lemma~\ref{lem:logn1}, second inequality comes from the fact that we consider an $r$-approximation algorithm, third inequality comes from Lemma~\ref{lem:logn2}, and the final equality comes from Lemma~\ref{lem:logn3}.
Hence, we obtained an $r \left(1+\frac{1}{|S|^{\lambda}}\right)$-approximation algorithm for the Minimum Set Cover problem, the ratio of which is lower than $\ln(n)$ for large enough $\lambda$.
However, as shown by Dinur and Steurer~\cite{dinur2014analytical}, it is not possible unless $P=NP$.
Therefore, there cannot exist such $r$-approximation algorithm for the Optimal Partial Diffusion Sequence problem.
This concludes the proof of Theorem~\ref{thrm:logn}.
\end{proof}

We now move to describing the conclusions and possible ideas for future work.

\section{Conclusions \& Future Work}
\label{sec:conclusions}

In this article we investigate the computational aspects of the strategic diffusion problem previously introduced by Alshamsi~\etal~\cite{alshamsi2018optimal}. We show that the partial diffusion problem is NP-complete in the general case, hence finding a polynomial algorithm is impossible, unless P=NP. Given this difficulty, we considered the problem from the parametrized complexity point of view and showed that the problem is fixed parameter-tractable when parametrized by the sum of the treewidth and maximum degree. On the negative side, we showed that two previously proposed heuristic algorithms for solving the problem, \ie, the greedy and the majority solutions, cannot have better than a logarithmic  approximation guarantee, even in the full diffusion case. Finally, we proved that the partial diffusion problem does not admit better than a logarithmic approximation, unless P=NP.

As for the potential ideas for future work, determining the complexity of the full diffusion problem is an interesting open question. Furthermore, our results concerning approximation algorithms are negative. Developing a polynomial algorithm with small approximation ratio would make the task of finding an efficient way of spreading the strategic diffusion much simpler. Another possible way of extending our work is showing effective algorithms finding the optimal solution in more restricted classes of networks, \eg, networks with bounded treewidth or bounded pathwidth (without the degree restriction).

The model of strategic diffusion presented in this work is based on a complex contagion process, \ie, multiple sources of exposure to the process greatly increase the probability of activation for a given node.
It would be possible to devise a similar model, but based on a simple contagion process, \eg, taking inspiration from the independent cascade model~\cite{kempe2003maximizing}.
Since simple contagion does not take multiple exposures into consideration, nodes activated earlier in the sequence would not affect the probability of activation in a given moment.
This could lead to optimal strategies significantly different than for the strategic diffusion based on complex contagion.

\section*{Acknowledgments}

This work was funded by the Cooperative Agreement between the Masdar Institute of Science and Technology (Masdar Institute), Abu Dhabi, UAE and the Massachusetts Institute of Technology (MIT), Cambridge, MA, USA---Reference 02/MI/MIT/CP/11/07633/GEN/G/00.
CAH and FLP acknowledge the support from the MIT Media Lab consortia.

\bibliographystyle{elsarticle-num}
\bibliography{bibliography-strdiff}

\end{document}